\documentclass[11pt]{article}

\usepackage[T1]{fontenc}
\usepackage{lmodern}
\usepackage{microtype}
\usepackage[margin=1.02in]{geometry}
\usepackage{amsmath,amssymb,amsthm,mathtools}
\usepackage{needspace}
\usepackage{authblk}
\usepackage[hidelinks]{hyperref}
\usepackage[nameinlink,noabbrev]{cleveref}

\usepackage[utf8]{inputenc}

\usepackage{graphicx}

\usepackage[noadjust]{cite}
\usepackage[
  convert-footnotes = true,
  placement = mixed
]{notes2bib}

\newtheorem{theorem}{Theorem}[section]
\newtheorem{lemma}[theorem]{Lemma}
\newtheorem{proposition}[theorem]{Proposition}
\newtheorem{corollary}[theorem]{Corollary}
\newtheorem{remark}[theorem]{Remark}
\newtheorem{definition}[theorem]{Definition}

\newcommand{\Cut}{\operatorname{c}}
\newcommand{\mc}{\operatorname{MC}}
\newcommand{\Eval}{\operatorname{Eval}}
\newcommand{\M}{\#\mc}
\newcommand{\FP}{\mathsf{FP}}
\newcommand{\ii}{\mathrm{i}}
\newcommand{\one}{\mathbf{1}}
\newcommand{\set}[1]{\left\{#1\right\}}

\title{Evaluating QAOA expectation values can be as hard as counting optimal solutions}
\author{Stuart Hadfield\thanks{\href{mailto:shadfield@usra.edu}{shadfield@usra.edu}}}
\affil{Quantum Artificial Intelligence Laboratory (QuAIL), NASA Ames Research Center, Moffett Field, CA 94035, USA}
\affil{USRA Research Institute for Advanced Computer Science, Mountain View, CA 94043, USA}

\date{\today}

\begin{document}
\maketitle

\begin{abstract}

Evaluating expectation values is a critical task for variational quantum eigensolvers,
parameterized quantum circuits, and many other quantum algorithms.
We consider the well-studied case of the Quantum Approximate Optimization Algorithm (QAOA) for the MaxCut problem.
Recent work of Wang et al. \href{https://arxiv.org/abs/2511.20212}{[arXiv:2511.20212]} showed this task to be NP-hard
in general for any QAOA depth $p\geq 2$, complementing past results showing efficiently computable formulas for $p=1$ with arbitrary problem graphs.
We sharpen this dichotomy showing that for
$p\geq2$ exact or exponentially precise
cost expectation value
evaluation is $\#$P-hard under deterministic polynomial-time Turing
reductions. Hardness at $p\geq 2$ is shown to remain even for evaluating
single pairwise correlators $\langle Z\otimes Z\rangle$,
as well as for
highly
restricted sets of algorithm parameters.
Our proof refines the NP-hardness construction of Wang et al. that recovers the maximum cut value from the largest exponent of a QAOA Laurent polynomial, utilizing a distinct and simpler construction that extracts a value proportional to the total number of maximum cuts, in addition to the optimal cut value.
Thus we show that
the QAOA expectation value hardness transition from
$p=1$ to $p=2$ is not only from tractability to optimization hardness, but to that of counting optimal solutions. As an application we show our results imply analogous hardness results for computing gradients and Hessians of QAOA circuits.
\end{abstract}

\section{Introduction}

Quantum computing has attracted much interest toward enabling novel approaches to NP-hard combinatorial optimization problems, in both exact and approximate optimization settings~\cite{abbas2023quantum}.
At the same time, rigorous results characterizing their power or limitations remain relatively sparse, impeding our assessment of their true
utility or potential for quantum advantage here. A popular paradigm is the
Quantum Approximate Optimization Algorithm~(QAOA)~\cite{Farhi2014,hadfield2019quantum}, in particular applied to the prototypical NP-hard MaxCut problem.
For analyzing quantum algorithms such as QAOA it is common to bound its expected approximation ratio in terms of the expectation value of a quantum observable encoding the cost function. This corresponds to the expected cost value of samples obtained from repeated preparation and measurement of a fixed QAOA state.
For QAOA at its lowest depth (number of layers) $p=1$ for MaxCut, exact analytic formulas for the expected cost value were obtained in \cite{Wang2018}. On the other hand, recent work of \cite{WangEtAl2025} has shown that computing the same quantities for $p\geq 2$ is
itself in general NP-hard
(i.e., can be as hard as the underlying problem itself). Our main contribution is to sharpen this dichotomy to show that the $p\geq 2$ case is in fact $\#$P-hard; i.e., the latter case is as hard as counting the number of optimal solutions of a constraint satisfaction problem, not just finding one, which is generally believed to be a harder task~\cite{Valiant1979,Krentel1988}.

\paragraph{QAOA for MaxCut}
Consider the MaxCut problem on a graph $G=(V,E)$ on $n=|V|$ vertices.
We identify cuts as spin (bit) assignments $x\in\{\pm1\}^{n}$ with cut value
\begin{equation}
 \Cut_G(x)=\sum_{\{u,v\}\in E}\frac{1-x_ux_v}{2},
 \quad
 \mc(G)=\max_x\Cut_G(x),
 \quad
 \M(G)=\set{x:\Cut_G(x)=\mc(G)}.
 \label{eq:cut-definitions}
\end{equation}
Finding a maximum cut is a canonical NP-hard optimization problem.
Counting maximum cuts exactly, i.e. determining $|\M(G)|$, is $\#$P-hard~\cite[Lemma~13]{JerrumSinclair1993} (see also~\cite{provan1983complexity,hermann2009complexity}).

For QAOA applied to MaxCut we define the standard cost and mixer Hamiltonians from the single-qubit Pauli matrices $X=(\begin{smallmatrix} 0 & 1 \\ 1 & 0 \end{smallmatrix} )$ and
 $Z=(\begin{smallmatrix} 1 & 0 \\ 0 & -1 \end{smallmatrix} )$
as
\begin{equation}
 C_G=\sum_{\{u,v\}\in E}\frac{I-Z_uZ_v}{2},
 \qquad
 B=\sum_{v\in V}X_v.
 \label{eq:hamiltonians}
\end{equation}
Here $C_G$ satisfies $C_G\lvert x\rangle=c_G(x)\lvert x\rangle$ for every computational basis state $\lvert x\rangle$.  At depth parameter~$p$, standard QAOA for MaxCut prepares the quantum state\footnote{Throughout, the QAOA depth parameter $p$ denotes the total number of
alternating cost and mixer rounds, rather than the physical gate
depth of an implementation.  Standard implementations of
$e^{-i\gamma C_G}$ require at least $\Delta(G)$ two-qubit gate depth. The hard instances we construct below
have $\Delta(G)=N-1$, so their physical
two-qubit gate depth grows with the instance size,  even when
$p=2$.}
\begin{equation}
\lvert\psi_p(\boldsymbol\gamma,\boldsymbol\beta)\rangle
  =(e^{-\ii\beta_p B}e^{-\ii\gamma_p C_G})\dots (e^{-\ii\beta_2 B}e^{-\ii\gamma_2 C_G})( e^{-\ii\beta_1 B}e^{-\ii\gamma_1 C_G})
 \lvert+\rangle^{\otimes n},
 \label{eq:qaoa-state}
\end{equation}
for easy-to-prepare initial product state $\lvert+\rangle^{\otimes n}=(\lvert0\rangle+\lvert1\rangle)^{\otimes n}/\sqrt{2^n}$,
while seeking $2p$ parameters (angles)
$\gamma_1,\dots,\gamma_p,\beta_1,\dots,\beta_p$ optimizing the cost expectation value
\begin{equation}
 F_p(G;\boldsymbol\gamma,\boldsymbol\beta)
 =\langle\psi_p\rvert C_G\lvert\psi_p\rangle.
 \label{eq:qaoa-objective}
\end{equation}
Given $F_p(G;\boldsymbol\gamma,\boldsymbol\beta)$ one can then obtain bounds on the expected approximation ratio obtained by repeatedly preparing and measuring the QAOA state
with the same parameters on a quantum computer~\cite{Farhi2014,Wang2018,hadfield2018quantum}. Evaluations of $F_p(G;\boldsymbol\gamma,\boldsymbol\beta)$
are also important for other algorithm aspects including
parameter optimization.

The complexity of evaluating Eq.~\eqref{eq:qaoa-objective} is distinct from the complexity of
sampling from the full QAOA output distribution~\cite{farhi2016quantum}.  At $p=1$, the expectation value has an explicit
formula on general graphs~\cite{Farhi2014,Wang2018} (see Appendix~\ref{sec:p1}).  By contrast, Wang, Chen, Sun,
and Zhang recently showed~\cite{WangEtAl2025} that exact evaluation is NP-hard for every fixed $p\geq2$. Their hardness proof writes the cost expectation as a Laurent polynomial in a phase
variable and recovers $\mc(G)$ from its largest nonzero exponent when evaluated for $C_{G'}$ on a suitably transformed graph $G'=G'(G)$.

The present work upgrades this mechanism for $p\geq 2$ from optimization to counting hardness.
At $p=2$, we tie the phase angles, $\gamma_1=\gamma_2=\phi$, and fix
the mixer angles to $\beta_1=\arctan(1/2)$ and $\beta_2=\pi/4$.
This leaves a one-parameter expectation value, which is a Laurent polynomial
in $z=e^{\ii\phi}$.
We construct reduced problem instances for
which the
extremal Laurent polynomial coefficient can be evaluated in closed form and is shown to be proportional to the full
maximum cut degeneracy $|\M(G)|$ of the input problem graph.
The required coefficient property is enforced by complete bipartite subgraph
$K_{L,L}$ variable gadgets with suitably chosen $L=O(n)$,
along with densely connected anchor variables.
See Figure~\ref{fig:triangle-gadget} below for an illustrative example of our construction.
Our main theorem uses polynomially many expectation value queries at different phase values to interpolate the required Laurent coefficient for $p=2$.
This interpolation argument applies to both exact evaluation and to evaluation
with exponentially small additive error.
Hence we show Turing reducibility between these problems, written informally as
\begin{eqnarray*}
    |\M| \; \leq_T\; \Eval^{\mathrm{MC}}_{p=2}.
\end{eqnarray*}
The general $p$ result then follows setting all angles in the additional $p-2$ layers to zero.
This yields $\#$P-hardness for computing QAOA expectation values rather than NP-hardness as in prior work. We further show that hardness remains when restricted to computing a single-edge
correlator~$\langle Z_rZ_s\rangle$.

\subsection{Main result: Evaluating QAOA $p=2$ expectation values is \#P-hard}

\begin{definition}[QAOA expectation value evaluation problem]
For fixed \(p \in \mathbb{N}\), let $\Eval_p^{\mathrm{MC}}$ be the function problem that maps a  graph $G$ and parameters $\boldsymbol\gamma,\boldsymbol\beta$ to the exact algebraic number $$F_p(G;\boldsymbol\gamma,\boldsymbol\beta).$$

The additive-error version additionally receives $b\in\mathbb{N}$ in unary\footnote{Here unary encoding of $b$ ensures that the requested number of accuracy bits
is part of the input length.  In particular, choosing \(b=O(N)\) permits
exponentially small error while preserving a polynomial-time reduction.}
and returns a rational number
$\widetilde F$ satisfying
$$\bigl|\widetilde F-
F_p(G;\boldsymbol{\gamma},\boldsymbol{\beta})\bigr|
\leq 2^{-b}.$$
\end{definition}

For a designated edge \(e=\{u,v\}\in E(G)\), we additionally define
\[
  \mathrm{Eval}^{\mathrm{uv}}_p
  (G;\boldsymbol{\gamma},\boldsymbol{\beta})
  :=
  \langle\psi_p|Z_uZ_v|\psi_p\rangle .
\]

We will primarily be concerned with deterministic polynomial-time Turing
reductions, i.e., what one can solve efficiently given a black box computing
$\Eval_p^{\mathrm{MC}}$.  The output of an exact evaluation oracle is understood to be expressed
in a standard basis of the compositum of the number fields in which the input
is represented.  The
particular numerical representation is inessential for the Turing reductions below,
provided it supports polynomial-time exact arithmetic and basis conversion.
Every oracle (subroutine) query uses a number field of polynomial degree and the inverse Fourier
transform is carried out exactly within that field.  QAOA angles are specified
exactly by algebraic values of their sines and cosines; the reduction uses only
roots of unity and fixed algebraic values for the sines and cosines of the mixer angles.\footnote{For the oracle queries in the proof of
Theorem~\ref{thm:main}, set $Q=4m+1$ and
$\zeta_Q=e^{2\pi\ii/Q}$.  Encode the phase $\phi_j=2\pi j/Q$ by
$(Q,j)$ with $e^{\ii\phi_j}=\zeta_Q^j$, and encode the fixed mixer pairs as
$(2/\sqrt5,1/\sqrt5)$ and $(1/\sqrt2,1/\sqrt2)$.  A fully explicit common
field is $\mathbb Q(\zeta_{\ell_Q})$, where
$\ell_Q=\operatorname{lcm}(Q,40)$; it contains
$\zeta_Q,\ii,\sqrt2,\sqrt5$ and has degree $\varphi(\ell_Q)=O(Q)$.
The path-sum expression bounds the heights of all basis coefficients by
$2^{\operatorname{poly}(N)}$.  Thus the exact answers have polynomial bit
length, and exact arithmetic, equality testing, basis conversion, and the
inverse discrete Fourier transform have polynomial bit
complexity~\cite{lenstra1982factoring,basu2006algorithms,cohen2013course}.}
The explicit number field representation is specified in
Appendix~\ref{app:single-query}.

Here $\FP$ denotes the class of functions computable by a deterministic
classical algorithm in time polynomial in the input length.

\begin{theorem}[QAOA expectation complexity dichotomy]\label{thm:main}
For unweighted simple graphs,
\begin{eqnarray}
 &&\Eval^{\mathrm{MC}}_1\in\FP,\\
 &&\Eval^{\mathrm{MC}}_p\text{ is }\#\text{P-hard for every fixed }p\geq2.
 \label{eq:threshold}
\end{eqnarray}
Moreover, for every fixed $p\geq2$,
it suffices to restrict evaluations to QAOA angles
\begin{equation}
 \gamma_1=\gamma_2,  \qquad
 \beta_1=\arctan \frac12,
 \qquad
 \beta_2=\frac{\pi}{4},
 \qquad
 \gamma_j=\beta_j=0\quad(3\leq j\leq p).
 \label{eq:base-restrictions}
\end{equation}
\end{theorem}
\begin{corollary}\label{cor:single-correlator}
    The same conclusions of the theorem hold if the oracle evaluates only one designated two-qubit correlator
$\langle Z_rZ_s\rangle$ rather than the global cost expectation value $\langle C_G \rangle $.
\end{corollary}

Here the mixer angles are fixed constants independent of the input, and
only the single tied cost phase angle
$\phi:=\gamma_1=\gamma_2$
is varied between oracle queries.

Appendix~\ref{app:single-query} strengthens the exact-evaluation statement to a
single oracle query. That reduction builds on the construction of
Theorem~\ref{thm:main} and requires additional algebraic steps to recover the
Laurent coefficients from one exact output. It also shows hardness for the symmetric
fixed choice $\beta_1=\beta_2=\pi/4$ of mixing angles after adding at most one isolated source
vertex. This single-query reduction is algebraic rather than necessarily numerically robust;
the polynomial interpolation using Fourier recovery in the main proof remains the stable route for
additive approximation.

We show the hardness dichotomy of the theorem extends directly to the setting of exponentially small additive error in Section~\ref{sec:expoApprox}.

\begin{remark}[Consequences of efficient exact evaluation]
 Suppose that, for  $p\ge 2$, the exact or exponentially precise expectation value of
the MaxCut QAOA circuit can be computed in polynomial
time on arbitrary graphs. Then by Toda's theorem~\cite{Toda1991} the polynomial hierarchy collapses to its lowest level, i.e. P=PH.
\end{remark}

For exponentially precise evaluation, we show below that, for
an $N$-vertex oracle instance, the extreme Laurent coefficient
has magnitude at least $2^{-cN}$ for a known constant $c$, and hence $O(N)$ bits of
additive precision suffice to recover the exact integer count
by rounding.

We emphasize that both our results and those of Ref.~\cite{WangEtAl2025} imply that there exist worst-case families of graphs for which evaluating $F_p$ appears computationally intractable, but this by no means implies this property applies to all graphs. Indeed, Ref.~\cite{WangEtAl2025} establishes complementary
algorithms parameterized by local treewidth~\cite{WangEtAl2025}, complementing
results obtained for other special cases for which rigorous results have been obtained
in the literature \cite{Farhi2014,Farhi2014b,Wang2018,hadfield2018quantum,jiang2017near,bravyi2020obstacles,hastings2019classical,marwaha2021local,wurtz2021maxcut,barak2021classical,marwaha2022bounds,hadfield2022analytical,ozaeta2022expectation,basso2021quantum,boulebnane2021predicting,farhi2022quantum,basso2022performance,boulebnane2025evidence}.
Furthermore, whether one can efficiently estimate a given QAOA cost expectation value on a classical computer or not, a quantum computer may still be required to obtain corresponding problem solution samples~\cite{farhi2016quantum}.

\subsection{Relationship to the results of \cite{WangEtAl2025} and related work}

Prior to \cite{WangEtAl2025}, a sequence of papers \cite{Farhi2014,Farhi2014b,Wang2018,hadfield2018quantum,jiang2017near,bravyi2020obstacles,hastings2019classical,marwaha2021local,wurtz2021maxcut,marwaha2022bounds,ozaeta2022expectation,hadfield2022analytical,basso2021quantum,boulebnane2021predicting,farhi2022quantum,basso2022performance,boulebnane2025evidence} obtained formulas or bounds for $\Eval_p^{\mathrm{MC}}$ in special cases.
Beyond performance analysis, evaluating QAOA expectation values is also critical for algorithm parameter selection; in \cite{bittel2021training} it was shown that the parameter optimization problem in some settings can itself become NP-hard.

For computing QAOA expectation values, Reference~\cite{WangEtAl2025} previously established
NP-hardness for every fixed $p\geq2$ by identifying the maximum cut value with the largest nonzero exponent
of a Laurent polynomial.
Here we retain the general Laurent polynomial approach
but replace the
underlying graph construction with one that utilizes a balanced counting gadget, and for this graph prove a corresponding exact leading  coefficient identity.
The new balanced gadgets make every maximum cut of the constructed $N$-vertex graph
have Hamming distance $N/2$ from each assignment maximizing the cut decrease
when the two anchors are flipped, so all such
cuts contribute the same known weight to the leading coefficient.
This property yields a count, rather than only the optimal value.
The resulting new conclusions are: (i) $\#$P-hardness through recovery of the
number of maximum cuts; (ii) hardness for a single edge correlator; and
(iii) the fixed-angle restrictions of Theorem~\ref{thm:main}. Our analysis is primarily concerned with the $p=2$ case. The $p\geq 2$ statement then follows by padding a
hard depth-two circuit with $p-2$ identity layers (i.e., layers for which $\gamma=\beta=0$).

Our results apply instead to computation of QAOA expectation values;
these quantities
can also be estimated \textit{quantumly} by repeated sampling from the quantum computer.
In \cite{farhi2016quantum}
Farhi and Harrow adapted the postselection-based framework of \cite{bremner2011classical}
to show that efficient classical sampling from the output of standard QAOA circuits
even at depth $p=1$
would imply a collapse of the polynomial hierarchy (to at least its third level),
with this result further sharpened in the recent paper~\cite{abolicnvs2026sharp}.
    Krovi later proved average-case hardness results~\cite{krovi2022average} for approximating output probabilities of random
    quantum circuits, including random $p=1$ QAOA ensembles.
In~\cite{dalzell2020many}
    the authors refine asymptotic supremacy arguments into a fine-grained framework, allowing
    quantitative statements about the scale required for classical intractability based on standard conjectures in complexity theory.
Hence in terms of sampling tasks, QAOA circuits provide a potential path to quantum computational supremacy~\cite{farhi2016quantum,harrow2017quantum,dalzell2020many}. Our complementary results show a stronger
conditional collapse of the polynomial hierarchy to PH=P if one could
efficiently evaluate QAOA expectation values (to exponential precision), but with distinct threshold $p=2$.

Our result should also be distinguished from work on more general
quantum mean value problems for constant-depth circuits
\cite{bravyi2021classical}.
In that setting, depth refers to physical circuit depth in a local
gate model, and substantially broader circuit and observable
families are considered.
By contrast the fixed parameter $p$ in
our theorems counts QAOA rounds.
As the graphs produced by our
reduction contain universal vertices, a standard two-qubit-gate
implementation of a cost layer has depth $\Omega(N)$; consequently,
our theorem is not a hardness result for constant physical-depth
circuits.
Its novelty is instead that exact $\#\mathrm{P}$-hardness
already occurs within the rigid family of standard fixed-round
MaxCut-QAOA circuits, with either the MaxCut Hamiltonian or one
designated two-qubit correlator as the observable.  Moreover, our
proof is qualitatively different, deriving hardness directly from
an extreme Laurent-polynomial coefficient rather than from
postselection-induced universality~\cite{bremner2011classical,bravyi2021classical}.  On the other hand, our theorem
concerns exact, or exponentially precise, evaluation and therefore
does not establish hardness in the constant-additive-error regime
considered in Ref.~\cite{bravyi2021classical}.

Further afield, Szegedy~\cite{szegedy2019qaoa} previously proposed viewing QAOA expectation values as graph invariants, and investigated whether structural information
such as graph non-isomorphism can be recovered from comparing these quantities.
Distinct from the expectation value setting, in \cite{hadfield2021representation} it was shown that
$\#\mathrm{P}$-hardness can arise in the construction of
Hamiltonian representations for compactly specified Boolean
functions. Complementary work establishes complexity barriers to efficiently preparing states
whose measured solutions recover a fixed positive fraction of the optimal
MaxCut gain over a random cut on every input~\cite{Hadfield2026StatePreparation}
and to optimizing Hamiltonians compressed by quantum random access optimization~\cite{Hadfield2026QRAOComplexity}. Here we study the complexity of evaluating expectation values
for specified QAOA circuits.
QAOA has also been proposed as a solution sampler
for approximate counting
\cite{DrapeauBanerjeeKourtis2025}.
Finally, we note that beyond \cite{WangEtAl2025}, Laurent polynomials have also proven useful in two recent works solving QAOA for MaxCut on the special case of \lq\lq ring of disagrees\rq\rq\ (two-regular graphs)~\cite{marwaha2026qaoa,kol2026machine}, confirming optimal expected cost scaling with $p$ originally conjectured in \cite{Farhi2014}.

The following section gives the details and proof of Theorem~\ref{thm:main}. We show a direct application of our results to assessing the hardness of computing gradients and Hessians for QAOA circuits in Sec.~\ref{sec:derivative-hardness}, and provide some additional discussion of our results in Section~\ref{sec:discussion}.

\par\addvspace{2ex}
\noindent\begin{minipage}{\linewidth}
\section{A balanced maximum-cut counting gadget}\label{sec:gadget}

To encode the maximum cuts of an input graph $H$, we replace each vertex
by a balanced complete bipartite gadget. In every maximum cut of the resulting
graph $G$, each gadget's two parts lie on opposite sides, and their orientations
specify a maximum cut of $H$ (Lemma~\ref{lem:rigidity}). Two anchor vertices,
joined to every other vertex, make their shared edge the sole contributor to
the largest Laurent exponent at depth two (Lemmas~\ref{lem:path-sum}
and~\ref{lem:boundary}). The balance ensures that every maximum cut contributes
the same weight to the corresponding coefficient, making it a known nonzero
multiple of $|\M(H)|$ at our chosen mixer angles (Theorem~\ref{thm:coefficient}).
Polynomial interpolation using Fourier recovery then yields the count
(Section~\ref{sec:hardness}).\par
\end{minipage}\par

\begin{figure}[h]
\centering
\includegraphics[width=0.85\textwidth]{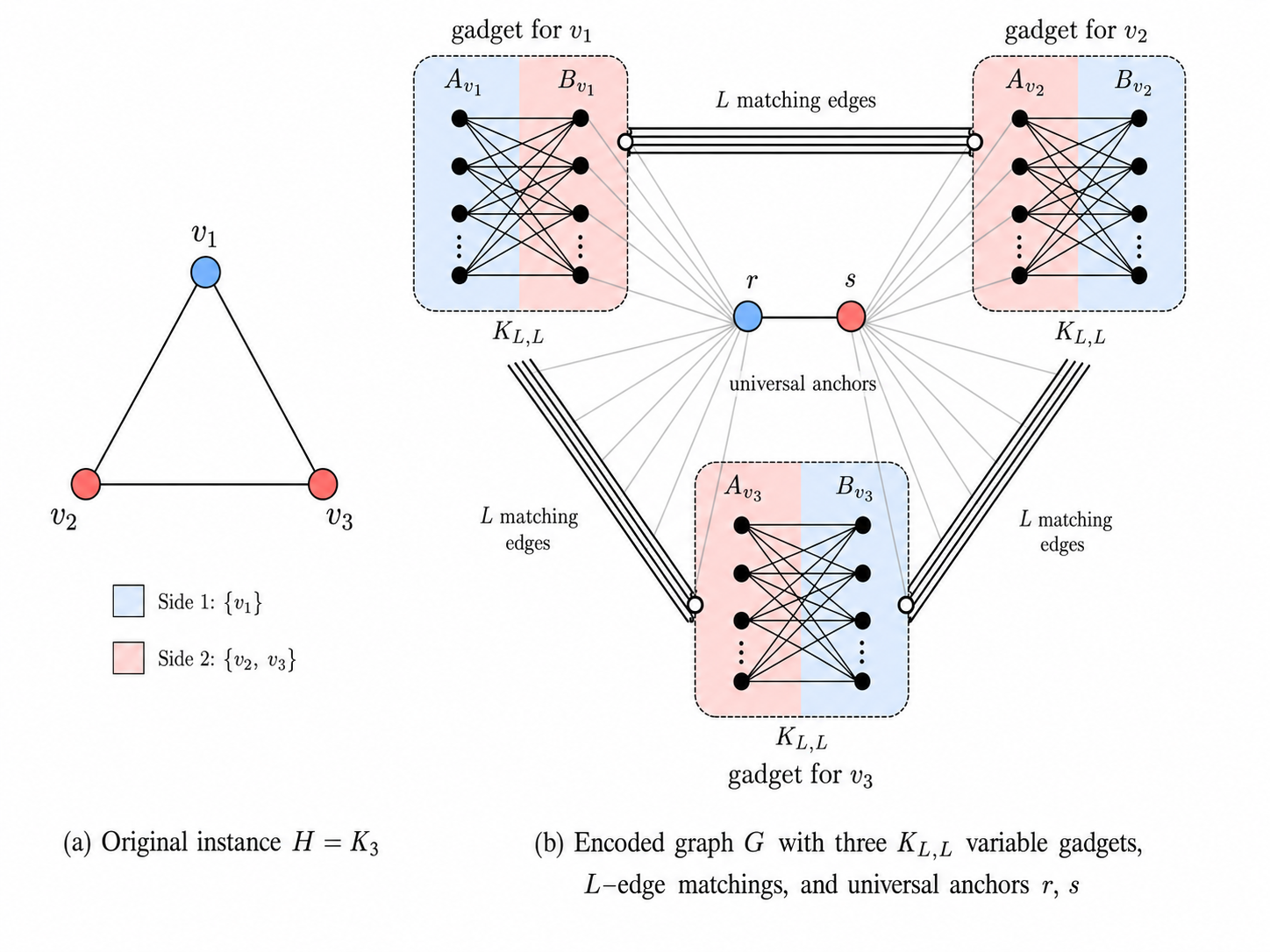}
\caption{The balanced counting gadget for the source instance $K_3$ (triangle graph).  Blue and red indicate the two labels of one source maximum cut and its corresponding lift.  Every $K_{L,L}$ variable
gadget is cut completely, so each $A_v$ and $B_v$ lie on opposite sides; the
synchronous matchings among the $A$-parts reproduce the three source edges.
The separated anchors $r$ and $s$ are adjacent to one another and each is
joined to every vertex in $W$.  Each part $A_v,B_v$
in panel~(b) contains $L$ vertices; only representative vertices and edges are shown.
For $H=K_3$, the constructed graph~$\mathcal G(H)$ has $L=9$, $N=56$, and maximum cut counts $|\M(G)|=2|\M(H)|=12$.
}
\label{fig:triangle-gadget}
\end{figure}

Let input graph $H=(V,E)$ have $n$ vertices and maximum degree $\Delta$.
Define the odd integer
\begin{equation}
 L=2(\Delta+2)+1.
 \label{eq:L}
\end{equation}
For each $v\in V( H)$, introduce disjoint vertex sets $A_v=\{a_{v,1},…,a_{v,L}\}$ and $B_v$, each of size $L$, and add all
edges of the complete bipartite graph between them.  For each
$\{u,w\}\in E(H)$ and each $j\in[L]$, add the synchronous edge
$\{a_{u,j},a_{w,j}\}$ connecting the corresponding vertices in $A_u$ and $A_w$.  (Synchronous edges are not added for vertices in the $B_v$.) Finally add two anchor vertices $r,s$, the edge $e_\star=\{r,s\}$, and
join each anchor to every vertex in
\begin{equation}
 W=\bigcup_{v\in V(H)}(A_v\cup B_v).
\end{equation}
Denote the resulting graph by $G=\mathcal G( H)$.  It is unweighted, simple, and
connected, with
\begin{equation}
 N:=|V(G)|=2 nL+2
 \label{eq:N}
\end{equation}
vertices. The number of edges is
\begin{equation}
  |E(G)|
  =
  nL^2+L|E(H)|+4nL+1.
\end{equation}
Since \(\Delta\le n-1\), we have \(L=O(n)\), and hence
 $N=|V(G)|=O(n^2)$ and
 $ |E(G)|=O(n^3)$.
Thus~\(\mathcal G(H)\), and all oracle queries generated from it, can be
constructed in polynomial time.

\smallskip
\par\Needspace{7\baselineskip}
\noindent The following lemma shows that every maximum cut of $G$ encodes a
maximum cut of $H$, and that each maximum cut of $H$ has exactly two lifts to
$G$.\par\nobreak

\begin{lemma}[Gadget rigidity]\label{lem:rigidity}
Every maximum cut of $G$ places each $A_v$ and $B_v$ monochromatically on opposite sides,
places $r$ and $s$ on opposite sides, and induces a maximum cut of $H$ through the
orientation of the sets $A_v$.  Conversely, every maximum cut of $H$ and either
orientation of $(r,s)$ determine a maximum cut of $G$.  Consequently,
\begin{align}
 \mc(G)&= nL^2+2 nL+1+L\mc( H),
 \label{eq:MCG}\\
 |\M(G)|&=2|\M(H)|.
 \label{eq:MG}
\end{align}
Moreover every maximum cut contains exactly $|W|/2= nL$ vertices of $W$ on each
side.
\end{lemma}

\begin{proof}
Each vertex of a variable $K_{L,L}$ has at most $\Delta+2$ external neighbors.
Consider one variable gadget with parts $A_v$ and $B_v$.
Its two proper orientations assign $+1$ to all of $A_v$ and $-1$ to all
of $B_v$, or the reverse. Choose the orientation at minimum Hamming distance
from the current labeling, allowing either choice in a tie.
Relative to this choice, let
$a$ vertices of $A_v$ and $b$ vertices of $B_v$ be misplaced, and set
$q=a+b$.  The Hamming distances from the current labeling of
\(A_v\cup B_v\) to the two proper bipartite orientations sum to \(2L\).
Therefore the nearer orientation has Hamming distance
\(q\le L\).
The
number of cut edges internal to the gadget is
\[
(L-a)(L-b)+ab=L^2-Lq+2ab.
\]
Thus correcting the orientation of the gadget gains at least
\[
Lq-2ab
   \geq Lq-\frac{q^2}{2}
   \geq \frac{Lq}{2}
\]
internal cut edges, where we used $2ab\leq q^2/2$ and $q\leq L$.
Each moved vertex has at most $\Delta+2$ external neighbors, so this
correction loses at most $(\Delta+2)q$ external cut edges.
This bound counts incidences between moved vertices and external edges;
an external edge incident to two moved vertices may therefore be counted
twice, but this is harmless for an upper bound.
Since
$L>2(\Delta+2)$, the net change is strictly positive whenever $q>0$.
Consequently every maximum cut orients each $K_{L,L}$ gadget
properly.
Thus each gadget
contributes $L^2$ cut edges and exactly $L$ vertices to each side.

If the anchors lie on the same side, flipping one anchor leaves the total number of cut
anchor--$W$ edges unchanged, because $W$ is balanced, and additionally cuts $\{r,s\}$.
Hence the anchors are opposite.  With opposite anchors, every $w\in W$ contributes exactly
one cut anchor edge.  The only orientation-dependent edges are the synchronous edges:
the $L$ edges representing $\{u,v\}\in E( H)$ are cut exactly when the orientations
of $A_u$ and $A_v$ differ.  The gadget and anchor edges contribute the fixed amount
$nL^2+2nL+1$, while the synchronous edges contribute $L$ times the encoded
cut value in $H$. Maximizing the cut of $G$ therefore maximizes this encoded
cut, giving~\eqref{eq:MCG}.
Finally, each labeled maximum cut of $H$ determines every gadget spin, and the two
anchor orientations give exactly two lifts, proving~\eqref{eq:MG}.
\end{proof}

Our graph construction is simpler and more symmetric than the gadget used in Ref.~\cite{WangEtAl2025}.  The latter employs asymmetric complete-bipartite vertex gadgets, a separate large complete-bipartite frame, and an additional controller vertex in order to isolate the maximal Laurent exponent.  By contrast, our reduction uses only balanced ($K_{L,L}$) variable gadgets and two universal anchors (r,s).

\subsection{The extreme coefficient at depth two}\label{sec:coefficient}

Fix an edge $e=\{u,v\}$ of a graph $G$.
Set $\gamma_1=\gamma_2=\phi$,
$\beta_1=\beta$, $\beta_2=\pi/4$, and define
\begin{equation}
  H_e(e^{\ii\phi},\beta)
 :=\langle\psi_2(\phi,\phi;\beta,\pi/4)|\,Z_uZ_v\,|
 \psi_2(\phi,\phi;\beta,\pi/4)\rangle.
 \label{eq:edge-correlator}
\end{equation}
For $a,b\in\{\pm1\}$ let
$ K_\beta^\pm(a,b)=\langle a|e^{\pm\ii\beta X}|b\rangle$.
Writing $c=\cos\beta$ and $s=\sin\beta$ we have
\begin{equation}
 K_\beta^\pm(a,b)=
 \begin{cases}c,&a=b,\\ \pm\ii s,&a\ne b,\end{cases}
 \label{eq:kernel-values}
\end{equation}
Let $T_e$ operate on an assignment string so as to flip the two vertex spins in $e$.
The next lemma rewrites the edge correlator as a sum of powers of the phase factor, which lets us analyze its coefficients.
The proof uses a sum-over-paths expansion over intermediate computational-basis
states, as in earlier analyses of QAOA circuits~\cite{basso2021quantum,WangEtAl2025,hadfield2022analytical}.
\begin{lemma}[Reduced path sum]\label{lem:path-sum}
For $z=e^{\ii\phi}$,
\begin{align}
 H_e(z;\beta)
 ={}&-2^{-N}\sum_{x,y,t\in\{\pm1\}^{V(G)}}
 y_uy_v\,
 z^{\Cut_G(x)-\Cut_G(t)+\Cut_G(y)-\Cut_G(T_ey)}
 \notag\\
 &\hspace{1.2cm}\times
 \prod_{w\in V(G)}K_\beta^+(x_w,y_w)
 K_\beta^-((T_ey)_w,t_w).
 \label{eq:path-sum}
\end{align}
Hence $H_e(z;\beta)$ is a Laurent polynomial in $z$.
\end{lemma}

\begin{proof}
Write $|s\rangle$ for the computational-basis state labeled by a spin
string $s\in\{\pm1\}^{V(G)}$, and insert $I=\sum_s|s\rangle\langle s|$
between successive operators in the QAOA expectation value. Label the strings at the first and second cost layers
on the bra (complex-conjugate) side by $x,y$, and on the ket side by
$t,\bar y$, respectively.
Here $\bar y$ is an independent label, not the spin complement of $y$.
Let $q$ label the basis state at the observable $Z_uZ_v$, whose eigenvalue
is $q_uq_v$. The diagonal cost layers supply the four cut-value phases below;
the two overlaps with the uniform initial state supply $2^{-N}$.
Since $B=\sum_w X_w$, the mixer matrix elements factor over vertices, giving

\begin{align*}
 H_e(z;\beta)={}&2^{-N}
 \sum_{x,y,\bar y,t,q\in\{\pm1\}^{V(G)}}q_uq_v\,
 z^{\Cut_G(x)-\Cut_G(t)+\Cut_G(y)-\Cut_G(\bar y)}\\
 &\quad\times\prod_{w\in V(G)}K_\beta^+(x_w,y_w)K_\beta^-(\bar y_w,t_w)
 \prod_{w\in V(G)}K_{\pi/4}^+(y_w,q_w)K_{\pi/4}^-(q_w,\bar y_w).
\end{align*}

Summing over the intermediate spins adjacent to the final mixer gives, for $w\notin e$,
\begin{equation}
 \sum_{q_w}K_{\pi/4}^+(y_w,q_w)K_{\pi/4}^-(q_w,\bar y_w)
 =\delta_{y_w,\bar y_w},
\end{equation}
and, for $w\in e$,
\begin{equation}
 \sum_{q_w}q_wK_{\pi/4}^+(y_w,q_w)K_{\pi/4}^-(q_w,\bar y_w)
 =-\ii y_w\delta_{\bar y_w,-y_w}.
\end{equation}
The two endpoint factors multiply to
$(-\ii y_u)(-\ii y_v)=-y_uy_v$. The Kronecker deltas enforce
$\bar y=T_ey$, giving~\eqref{eq:path-sum}.
\end{proof}

Define the largest possible cut-value decrease from flipping both endpoints of $e$:
\begin{equation}
 R_e=\max_y\bigl[\Cut_G(y)-\Cut_G(T_ey)\bigr].
 \label{eq:boundary}
\end{equation}

\par\Needspace{7\baselineskip}
\noindent The next lemma determines the largest possible decrease in cut value
caused by flipping both endpoints of an edge. This bounds the term
$\Cut_G(y)-\Cut_G(T_ey)$ in the correlator's Laurent exponents.\par\nobreak

\begin{lemma}[Largest cut decrease]\label{lem:boundary}
For every edge $e=\{u,v\}$ of a simple graph, let $d_u$ and $d_v$
denote the degrees of $u$ and $v$, respectively. Then
\begin{equation}
 R_e=d_u+d_v-2.
 \label{eq:boundary-value}
\end{equation}
\end{lemma}

\begin{proof}
Only edges incident on exactly one endpoint can change when both endpoints are flipped, so
$R_e\leq d_u+d_v-2$.  Equality is attained by assigning the same spin to $u,v$ and the
opposite spin to every other neighbor of either endpoint. This assignment is consistent even when \(u\) and \(v\) have common
neighbors, since every vertex in
\((N(u)\cup N(v))\setminus\{u,v\}\) receives the same spin opposite
to \(u\) and \(v\).
\end{proof}

For the anchor edge $e_\star=\{r,s\}$ of $G=\mathcal G(H)$, both endpoints are
universal and therefore
\begin{equation}
 R_{e_\star}=2N-4.
 \label{eq:Rstar}
\end{equation}
The only assignments attaining $R_{e_\star}$ are
\begin{equation}
 (q_\sigma)_r=(q_\sigma)_s=\sigma,
 \qquad
 (q_\sigma)_w=-\sigma\quad(w\in W),
 \qquad \sigma\in\{\pm1\},
 \label{eq:q-sigma}
\end{equation}
for which $T_{e_\star}q_\sigma=-\sigma\one$.  Since every nonanchor vertex has degree
strictly less than $N-1$, $e_\star$ uniquely maximizes $R_e$ among all edges.
Using $c_G(x)\le\mc(G)$, $c_G(t)\ge0$, and $R_e\le R_{e_\star}$,
the exponent in~\eqref{eq:path-sum} satisfies

\[
c_G(x)-c_G(t)+\bigl(c_G(y)-c_G(T_ey)\bigr)
\le \mc(G)+R_e
\le \mc(G)+2N-4.
\]

Hence every exponent in \eqref{eq:path-sum} is at most
\begin{equation}
 D_\star=\mc(G)+2N-4.
 \label{eq:Dstar}
\end{equation}
Since each summand is bounded independently,  attaining the upper bound requires $c_G(x)=\mathrm{MC}(G)$, $c_G(t)=0$, and $c_G(y)-c_G(T_{e_\star}y)=R_{e_\star}$.  Since
$G=\mathcal G(H)$ is connected, its only zero-cut assignments are the two
constant assignments.
Hence
at equality, necessarily
\begin{equation}
 x\in\M(G),
 \qquad
  t=\tau\one,\,\,\tau \in \{\pm 1\},
 \qquad
 y\in\{q_+,q_-\}.
 \label{eq:extreme-histories}
\end{equation}
By Lemma~\ref{lem:rigidity}, each $x\in\M(G)$ is balanced on $W$ and separates the anchors.
Consequently
\begin{equation}
 d_{Hamm}(x,q_\sigma)=1+\frac{|W|}{2}=\frac N2,  \label{eq:hamming}
\end{equation}
independently of $x$ and $\sigma$.

\par\Needspace{7\baselineskip}
\noindent At the chosen mixer angle, the coefficient of the highest power present
in the Laurent polynomial determines the number of maximum cuts.\par\nobreak

\begin{theorem}[Extreme-coefficient identity]\label{thm:coefficient}
Recall $D_\star := \mathrm{MC}(G)+2N-4$ as in Eq.~\eqref{eq:Dstar}. Write
\begin{eqnarray}
    H_{e_\star}(z;\beta)=\sum_{k\in\mathbb Z}a_k(\beta)z^k.
\end{eqnarray}
For \(G=\mathcal G(H)\), we have
  $a_k(\beta)=0$ for every $k>D_\star$,
and the coefficient at \(D_\star\) is
\begin{equation}
 a_{D_\star}(\beta)
 =-2^{1-N}|\M(G)|\,
 c^{N/2}(\ii s)^{N/2}\left[c^N+(-\ii s)^N\right],
 \label{eq:edge-top-coefficient}
\end{equation}
where \(c=\cos\beta\) and \(s=\sin\beta\). In particular, if \(\tan\beta=1/2\), then
\(a_{D_\star}(\beta)\neq0\), and so
\[
  D_\star=\max\{k\in\mathbb Z:a_k(\beta)\neq0\}.
\]
Hence, for this mixer angle \(|\#\mathrm{MC}(G)|\) is recovered by division
by a known nonzero algebraic factor.
\end{theorem}

\begin{proof}
Only the paths in Eq.~\eqref{eq:path-sum} satisfying the equality conditions in
Eq.~\eqref{eq:extreme-histories} can contribute to~$z^{D_\star}$.
The exponent in \eqref{eq:path-sum} attains $D_\star = \mathrm{MC}(G) + 2N-4$
if and only if $c_G(x) = \mathrm{MC}(G)$, $c_G(t) = 0$, and
$c_G(y) -c_G(T_{e_\star} y) = 2N-4$ simultaneously.
Fix such a path $(x, q_\sigma, \tau\mathbf{1})$ with
$x \in \#\mathrm{MC}(G)$ and $\sigma, \tau \in \{\pm 1\}$.
First observe the endpoint factor is
$y_r y_s = (q_\sigma)_r (q_\sigma)_s = \sigma^2 = 1$ for both $\sigma$.
Next, by \eqref{eq:hamming} the Hamming distance
$h = d_{\mathrm{Hamm}}(x, q_\sigma) = N/2$ is the same integer for every
$x \in \#\mathrm{MC}(G)$ and both $\sigma$, so by \eqref{eq:kernel-values} we have
\begin{equation}
\prod_{w \in V(G)} K^{+}_{\beta}\bigl(x_w, (q_\sigma)_w\bigr)
= c^{\,N-h}(\mathrm{i}s)^{h}
= c^{N/2}(\mathrm{i}s)^{N/2},
\label{eq:Kplus-product}
\end{equation}
independently of $x$ and $\sigma$.
Finally, by \eqref{eq:extreme-histories} we have $t_w = \tau$ for all $w$, and since $T_{e_\star} q_\sigma = -\sigma\mathbf{1}$,
\begin{equation}
\prod_{w \in V(G)} K^{-}_{\beta}\bigl((T_{e_\star} q_\sigma)_w,\, t_w\bigr)
= K^{-}_{\beta}(-\sigma,\tau)^{N}
=
\begin{cases}
c^{N} & \tau = -\sigma\\[2pt]
(-\mathrm{i}s)^{N} & \tau = \sigma.
\end{cases}
\label{eq:Kminus-product}
\end{equation}
Summing \eqref{eq:Kminus-product} over $\tau \in \{\pm 1\}$ gives
$c^{N} + (-\mathrm{i}s)^{N}$, independently of $\sigma$.
Since the endpoint factor, the product \eqref{eq:Kplus-product}, and the
$\tau$-sum are all independent of $x$ and $\sigma$, the triple sum over the
surviving paths factorizes as
\begin{align}
a_{D_\star}(\beta)
&= -2^{-N} \sum_{x \in \#\mathrm{MC}(G)} \; \sum_{\sigma \in \{\pm 1\}}
   \; \sum_{\tau \in \{\pm 1\}}
   y_r y_s \prod_{w} K^{+}_{\beta}\bigl(x_w,(q_\sigma)_w\bigr)
   \, K^{-}_{\beta}\bigl((T_{e_\star} q_\sigma)_w, t_w\bigr) \notag\\
&= -2^{-N} \cdot \underbrace{|\#\mathrm{MC}(G)|}_{x\text{-sum}}
   \cdot \underbrace{2}_{\sigma\text{-sum}}
   \cdot \underbrace{1}_{y_r y_s}
   \cdot \underbrace{c^{N/2}(\mathrm{i}s)^{N/2}}_{\eqref{eq:Kplus-product}}
   \cdot \underbrace{\bigl[c^{N} + (-\mathrm{i}s)^{N}\bigr]}_{\tau\text{-sum}}
   \notag, \label{eq:extreme-coeff-derivation}
\end{align}
which rearranges to give Eq.~\eqref{eq:edge-top-coefficient}.
For $\tan\beta=1/2$, the two terms in the bracket always have unequal
magnitudes and cannot cancel, which implies $a_{D_\star}(\beta)\neq 0.$
\end{proof}

For the global expectation,
\begin{equation}
 \mathcal{F}_2(G;z,\beta):=F_2(G;\phi,\phi,\beta,\pi/4)=\frac{|E(G)|}{2}-\frac12\sum_{e\in E(G)}H_e(z;\beta).
 \label{eq:global-sum}
\end{equation}
Since $e_\star$ uniquely maximizes $R_e$, the coefficient of $z^{D_\star}$ in
\eqref{eq:global-sum} is exactly $-a_{D_\star}/2$.  Thus either the single correlator or
the global objective contains the same count information.

\subsection{Coefficient recovery and proof of the main theorem}\label{sec:hardness}

Let $m=|E(G)|$.  Every Laurent exponent in \eqref{eq:path-sum} lies in $[-2m,2m]$.  Set
\begin{equation}
 Q=4m+1,
 \qquad
 \omega=e^{2\pi\ii/Q},
 \qquad
 \phi_j=\frac{2\pi j}{Q}.
 \label{eq:dft-points}
\end{equation}
For either $P(z)=H_{e_\star}(z;\beta)$ or $P(z)=\mathcal F_2(G;z,\beta)$,
query the oracle with $\gamma_1=\gamma_2=\phi_j$ for $j=0,1,\dots,Q-1$, so
$z=e^{\ii\phi_j}=\omega^j$. The exact values $P(\omega^j)$
recover every coefficient by the inverse discrete Fourier transform
\begin{equation}
 a_k(\beta)=[z^k]P(z)=\frac1Q\sum_{j=0}^{Q-1}P(\omega^j)\omega^{-jk}.
 \label{eq:idft}
\end{equation}
As the Laurent exponent interval \([-2m,2m]\) has width
\(4m<Q\), two distinct exponents in this interval are never congruent
modulo \(Q\).  Hence the inverse discrete Fourier transform does not suffer from aliasing and recovers the ordinary Laurent coefficients, including
those with negative indices.

\begin{proof}[Proof of Theorem~\ref{thm:main}]
The $p=1$ statement is Proposition~\ref{prop:p1}.  For $p=2$, start from an instance $H$ of counting
maximum cuts.
Set
\begin{equation}
 \cos\beta_1=\tfrac2{\sqrt5},
 \qquad
 \sin\beta_1=\tfrac1{\sqrt5},
\end{equation}
which satisfies $\tan \beta_1=1/2$.
The coefficient \eqref{eq:edge-top-coefficient} is nonzero for every even $N$ and directly
determines $|\M(G)|=2|\M(H)|$.
All coefficients are obtained from Eq.~\eqref{eq:idft}, and the largest exponent with nonzero coefficient is selected. Theorem \ref{thm:coefficient} guarantees that this exponent is $D_\star$
and that its coefficient is nonzero.
Thus exact evaluation is $\#$P-hard at QAOA depth two for either the global expectation or the designated correlator.

For any fixed $p>2$, set all angles in layers $3,\ldots,p$ to zero.  These layers are the
identity, so the depth-$p$ oracle value equals the depth-two value.  This proves the
$\#$P-hardness statement for every fixed $p\geq2$.
\end{proof}

\begin{proof}[Proof of Corollary~\ref{cor:single-correlator}]
The depth-one correlator formula is given in Appendix~\ref{sec:p1}.
For $p=2$, the proof above recovers the extreme coefficient directly from
$H_{e_\star}$, so no other correlator is needed.
The same identity padding applies for $p>2$.
\end{proof}

\begin{remark}
    The reduction of Theorem~\ref{thm:main} recovers more than the count of maximum cuts: the extreme exponent $D_\star = MC(G) + 2N - 4$ itself reveals $MC(G)$, and hence $MC(H)$ via Eq.~\eqref{eq:MCG}, before the coefficient is examined. Thus a single run of the interpolation recovers both the maximum cut value and the number of maximum cuts of H, so the present construction also reproduces the NP-hardness conclusions of \cite{WangEtAl2025} directly.
\end{remark}

\paragraph{Hardness on diameter-two graphs}

\begin{corollary}\label{cor:diameter-two-hardness}
For every fixed \(p \ge 2\), exact evaluation of the MaxCut QAOA
expectation is \(\#\mathrm{P}\)-hard under deterministic
polynomial-time Turing reductions, even when restricted to connected,
unweighted, simple graphs of diameter two. The same conclusion holds
for the expectation of the designated two-qubit observable \(Z_r Z_s\),
where \(\{r,s\}\) is an edge joining two universal vertices.
\end{corollary}

\begin{proof}The graph \(\mathcal G(H)\) constructed in the proof of
Theorem~\ref{thm:main} is unweighted, simple, and connected. The two
anchor vertices \(r\) and \(s\) are adjacent to one another and to
every vertex in every variable gadget.
Let \(u,v \in V(\mathcal G(H))\). If neither vertex is \(r\), then both are
adjacent to \(r\), and hence there is a path \(u-r-v\) of length at
most two. If one of the vertices is \(r\), then the pair is adjacent.
Thus every pair of distinct
vertices is joined by a path of length at most two, so
$ \operatorname{diam}(\mathcal G(H)) \le 2$.
The graph is not complete: for example, two distinct vertices in the
same part \(A_v\) of a variable gadget \(K_{L,L}\) are nonadjacent.
Hence
$\operatorname{diam}(\mathcal G(H)) = 2$.

\end{proof}

\begin{remark}
The diameter-two restriction is obtained at the expense of large
degree that scales with the problem size, as it must: any diameter-two graph on $N$ vertices has
maximum degree at least $\sqrt{N-1}$ (because from any vertex, at most \(1+\Delta+\Delta(\Delta-1)
=\Delta^2+1\) vertices can lie within distance two).
Each anchor is adjacent to every gadget vertex, and the
vertices within each \(K_{L,L}\) gadget have degree growing with the
source-instance
input problem size. Consequently, the corollary does not establish
hardness for bounded-degree graph families.
\end{remark}

\subsection{Hardness of exponentially precise additive approximation} \label{sec:expoApprox}
Recalling the setting of Theorem~\ref{thm:main} let
\begin{eqnarray}
      F_G(\phi)
  :=
  \langle \psi_p(\phi) | C_G | \psi_p(\phi) \rangle,
\end{eqnarray}
where $|\psi_p(\phi)\rangle$ is the single-parameter depth-$p$ QAOA state with $\gamma_1=\gamma_2=\phi$, $\beta_1=\textrm{arctan}(1/2), \beta_2=\pi/4$, and all other angles fixed $\gamma_j=\beta_j=0$.
We show that exponentially small evaluation errors still allow exact recovery of the number of optimal cuts.

\begin{corollary}[Exponentially precise additive evaluation]\label{cor:precision}
Fix \(p\ge2\) and the mixer angle choices of Theorem~\ref{thm:main}.
There exists a constant \(\alpha>0\) such that the following problem is
$\#$P-hard under deterministic polynomial-time Turing reductions: given an
\(N\)-vertex MaxCut instance from the family constructed in the proof and a
tied phase angle specified by the exact algebraic input $z=e^{\ii\phi}$, return a
rational number \(\widetilde F_G(\phi)\) satisfying
\[
   \bigl|\widetilde F_G(\phi)-F_G(\phi)\bigr|
   \le 2^{-\alpha N}.
\]
The same statement holds for the designated correlator
\(\langle Z_rZ_s\rangle\).
\end{corollary}

\begin{proof}
Write
\[
   F_G(z)=\sum_{k=-D}^{D}a_kz^k,
   \qquad z=e^{i\phi}.
\]
Evaluate the approximation oracle at \(M>2D\) distinct \(M\)-th
roots of unity and apply the normalized inverse discrete Fourier
transform.  If every oracle value has additive error at most
\(\varepsilon\), then every recovered Laurent coefficient has
additive error at most \(\varepsilon\).  Indeed, each coefficient error is an average of the oracle errors
multiplied by phase factors of unit magnitude. Its magnitude is therefore at
most $M^{-1}\sum_{j=0}^{M-1}|\delta_j|\le\varepsilon$ by the triangle
inequality, where $\delta_j$ is the error in the $j$th oracle value.
For the mixer angle choice $\tan\beta=1/2$, write
$c_\beta=\cos\beta=\frac{2}{\sqrt{5}}$ and
$s_\beta=\sin\beta=\frac{1}{\sqrt{5}}$, and define

\[
\kappa_N:=-2^{1-N}c_\beta^{N/2}(\ii s_\beta)^{N/2}
\left[c_\beta^N+(-\ii s_\beta)^N\right].
\]

This is the factor multiplying $|\M(G)|$ in~\eqref{eq:edge-top-coefficient}.

By Eq.~\eqref{eq:global-sum} and Theorem~\ref{thm:coefficient}, for the coefficient of $z^{D_\star}$ in \(F_G(z)\) we may write
\[
a_{D_\star} = -\tfrac12\kappa_N |\M(G)| = -\kappa_N |\M(H)|.
\]

Then $|\kappa_N|$ is bounded
from below (cf. Eq.~\eqref{eq:edge-top-coefficient}) by
\begin{align}
2^{1-N}(c_\beta s_\beta)^{N/2}
 \left|c_\beta^N+(-is_\beta)^N\right|
\nonumber
\,\geq&\,\,
2^{1-N}\left(\tfrac{2}{5}\right)^{N/2}
 \left(c_\beta^N-s_\beta^N\right)
\nonumber\\
\,=&\,\,
2^{1-N}\left(\tfrac{2}{5}\right)^{N/2}
 \tfrac{2^N-1}{5^{N/2}}
\nonumber\\
\,\geq&\,\,
\left(\tfrac{\sqrt{2}}{5}\right)^N\,
=
2^{-N\log_2(5/\sqrt{2})}.
\end{align}
Here we used the triangle inequality and
$2^N-1\geq 2^{N-1}$.
Hence $|\kappa_N|\ge 2^{-\alpha_0 N}$ for $\alpha_0:=\log_2(5/\sqrt2)$.

Set $b_N:=\lceil\alpha_0N\rceil+2$. Any oracle request with precision
$b\ge b_N$ has absolute error at most
\[
 \varepsilon\le 2^{-b}
 \le 2^{-b_N}
 \le \tfrac14 2^{-\alpha_0N}
 \le |\kappa_N|/4.
\]
The two cases are separated from the threshold $|\kappa_N|$ by at least
$3|\kappa_N|/4$, so the magnitude comparisons below can be made using a further
$O(N)$ bits of ordinary numerical precision.
Let $\widehat a_k$ denote the recovered coefficients.  For $k>D_\star$ the
true coefficient vanishes, so
$|\widehat a_k|\le\varepsilon<|\kappa_N|$.  Using
$|\#\mathrm{MC}(H)|\ge2$ (global spin reversal preserves the cut value),
\[
 |\widehat a_{D_\star}|
 \ge |\kappa_N|\,|\#\mathrm{MC}(H)|-\varepsilon
 \ge \tfrac74|\kappa_N|>|\kappa_N|.
\]
Hence $D_\star$ is the largest exponent whose recovered coefficient exceeds
$|\kappa_N|$ in magnitude, a threshold computable from $N$ and $\beta$.
Dividing $\widehat a_{D_\star}$ by the known nonzero factor $-\kappa_N$ yields
$|\#\mathrm{MC}(H)|$ to additive error at most $1/4<1/2$, so rounding to the
unique nearest integer recovers the counts $|\#\mathrm{MC}(G)|$ and
$|\#\mathrm{MC}(H)|$ exactly.

The same arguments apply for the single correlator
\(\langle Z_rZ_s\rangle\). In that case Theorem~\ref{thm:coefficient} gives the coefficient directly without the \(-\tfrac12\) factor from Eq.~\eqref{eq:global-sum}.
 Thus the same
threshold and rounding argument recovers $|\#\mathrm{MC}(G)|$ and $|\#\mathrm{MC}(H)|$ directly; only
the known multiplicative factor relating the extreme coefficient to
the count is changed.

Finally, choose a fixed integer $\alpha>\alpha_0+1$.
Since every graph in the constructed family has $N\ge2$,
$b_N\le (\alpha-1)N+2\le\alpha N$.
Requesting $b=\alpha N$ in unary therefore gives the stated error bound
$2^{-\alpha N}$ with $O(N)$ accuracy bits, and the recovery argument above applies.

\end{proof}

\section{Hardness of gradients and Hessians}
\label{sec:derivative-hardness}

Consider the setting of Theorem~\ref{thm:main} where we restrict to the
tied-phase line
\[
  F_G(\phi)
  :=
  \left\langle C_G\right\rangle_{\gamma_1=\gamma_2=\phi}.
\]
The same statements hold with $C_G$ replaced by the designated observable
$Z_rZ_s$.  As above,
\begin{equation}
  F_G(\phi)=\sum_{k=-D}^{D} a_ke^{ik\phi},
  \label{eq:compact-derivative-laurent}
\end{equation}
where $D\leq 2|E(G)|=\operatorname{poly}(|G|)$.  For the graph $\mathcal G(H)$ produced by the
reduction we again use that in either case
by Theorem~\ref{thm:coefficient} there is a $D_\star>0$ such that
  $a_{D_\star}(G)
  =
  \kappa(H)\,|\M(H)|$, with
  $\kappa(H)\neq 0$ computable from \(N\) and the chosen
observable. The exponent \(D_\star\) itself need not be known in
advance, since it contains \(\mathrm{MC}(H)\).
Differentiating $F_G(\phi)$ preserves the count encoded in its
extreme coefficient, giving the same hardness for derivatives along the tied cost phase.

\begin{theorem}[Hardness of phase derivatives]
\label{thm:compact-derivative-hardness}
For every fixed $p\ge2$ and every fixed integer $r\ge1$, exact computation
of
\[
  F_G^{(r)}(\phi)=\frac{d^r}{d\phi^r}F_G(\phi)
\]
is $\#\mathrm{P}$-hard under deterministic polynomial-time Turing
reductions.  This holds for unweighted simple connected MaxCut instances
for the parameter choices of  Theorem~\ref{thm:main}, and for both the
total MaxCut expectation and the correlator $\langle Z_rZ_s\rangle$.
\end{theorem}

To the best of our knowledge we are not aware of an
earlier \(\#\mathrm P\)-hardness result for exact derivatives of the
standard fixed-depth QAOA for MaxCut objective.

\begin{proof} Differentiating \eqref{eq:compact-derivative-laurent} gives
\[
  F_G^{(r)}(\phi)
  =
  \sum_{k=-D}^{D}(ik)^r a_ke^{ik\phi}.
\]
Choose $M>2D$, let $\omega=e^{2\pi i/M}$, and set
$\phi_j=2\pi j/M$.  Orthogonality of the $M$th roots of unity yields
\[
  \frac1M\sum_{j=0}^{M-1}
  F_G^{(r)}(\phi_j)\omega^{-jk}
  =
  (ik)^r a_k.
\]
Since $D_\star>0$ and \(a_{D_\star}\ne0\),  the derivative
coefficient $ (iD_\star)^r a_{D_\star}$ is nonzero. Consequently, the largest exponent with nonzero recovered
derivative coefficient is again \(D_\star\); no oracle for \(F_G\)
itself is required.
The extreme coefficient is recovered exactly as
\[
  a_{D_\star}(G)
  =
  \frac{1}{(iD_\star)^rM}
  \sum_{j=0}^{M-1}
  F_G^{(r)}(\phi_j)\omega^{-jD_\star},
\]
where the value $D_\star$ is recovered as in the proof of Thm.~\ref{thm:main}.
Applying this identity to $\mathcal G(H)$ and using
Theorem~\ref{thm:coefficient}
recovers
$|\M(H)|$ in polynomial time.
\end{proof}

\par\Needspace{7\baselineskip}
\noindent The full gradient and Hessian determine the phase derivatives, so their exact computation inherits the preceding hardness.\par\nobreak

\begin{corollary}[Hardness of the full gradient and Hessian]
\label{cor:compact-gradient-hessian}
Consider the depth-two QAOA for MaxCut objective
 $ \mathcal F_G(\gamma_1,\gamma_2,\beta_1,\beta_2)
  :=\langle C_G\rangle$.
 Exact computation of the full gradient
$\nabla\mathcal F_G$ and of the full Hessian $\nabla^2\mathcal F_G$ is
$\#\mathrm{P}$-hard under deterministic polynomial-time Turing reductions.
The same conclusion holds for single correlators.
\end{corollary}

\begin{proof}[Proof of Corollary \ref{cor:compact-gradient-hessian}]
On the affine line
$\theta(\phi)=(\phi,\phi,\beta_1,\beta_2)$, the chain rule gives
\[
  F_G'(\phi)
  =
  v_\gamma^{\mathsf T}\nabla\mathcal F_G(\theta(\phi)),
  \qquad
  F_G''(\phi)
  =
  v_\gamma^{\mathsf T}\nabla^2\mathcal F_G(\theta(\phi))v_\gamma,
\]
for $v_\gamma=(1,1,0,0)^T$.
The conclusion follows from
Theorem~\ref{thm:compact-derivative-hardness} with $r=1$ and $r=2$.
\end{proof}

\begin{remark}
The result establishes hardness of the tied-phase directional derivative
and curvature, and hence of returning the full gradient or Hessian.  It
does not by itself imply hardness of each individual gradient coordinate
or Hessian entry separately.
\end{remark}

\par\Needspace{7\baselineskip}
\noindent The derivative hardness also survives exponentially small additive errors, by the same coefficient recovery argument used for expectation values.\par\nobreak

\begin{corollary}[Exponentially precise derivative approximation]
\label{cor:approximate-derivative-hardness}
For the mixer angles of Theorem~\ref{thm:main} and every fixed $r\ge1$,
there exists a constant $c>0$ such that approximating
$F_G^{(r)}(\phi)$ to additive error at most $2^{-cN}$, where
$N=|V(G)|$, is $\#\mathrm{P}$-hard under deterministic polynomial-time
Turing reductions.  In particular, this applies to the tied-phase first
and second derivatives.
\end{corollary}

\begin{proof}

The Fourier transform of the derivative oracle yields approximations to
\((ik)^r a_k\). As in Corollary~\ref{cor:precision}, for every \(k>D_\star\) the true
coefficient vanishes, whereas the coefficient at \(k=D_\star\) exceeds a
known computable threshold in magnitude. Scanning the recovered
coefficients from the largest exponent downward therefore identifies
\(D_\star\).

Suppose that each derivative query is returned with additive error at most
$\varepsilon$.  The normalized inverse Fourier transform shows that the
recovered value of $(ik)^ra_k$ has additive error at most
$\varepsilon$.  Consequently, the recovered value of $a_{D_\star}(G)$ has
error at most
$\varepsilon/D_\star^r$.
On the constructed family,
\[
  a_{D_\star}(\mathcal G(H))
  =
  \kappa(H)\,|\M(H)|,
\]
with $|\kappa(H)|=2^{-O(N)}$ as shown.
It is therefore enough to require
\(\varepsilon/D_\star^r<|\kappa(H)|/2\).
Because $D_\star=\operatorname{poly}(N)$ and $r$ is fixed, this inequality
holds for $\varepsilon\le2^{-cN}$ with a suitable constant~$c>0$.
Division by the known factor $\kappa(H)$ followed by rounding recovers the
exact count.
\end{proof}

\section{Discussion} \label{sec:discussion}

Our main theorem gives a clean exact-complexity classification by depth for the standard QAOA for MaxCut
expectation values on unrestricted graphs: $p=1$ is polynomial-time computable, whereas every
fixed $p\geq2$ is already $\#$P-hard.  Our reduction uses high-degree anchors, consistent
with the tractability of fixed-depth cost expectation value evaluation on bounded-degree or bounded-local-treewidth
graph families~\cite{WangEtAl2025}.
This degree dependence differs sharply from
that of the full-distribution sampling problem: for general
$2$-local cost Hamiltonians, \={A}boli\c{n}\v{s} and
Ambainis~\cite{abolicnvs2026sharp} recently showed that depth-one QAOA
sampling is already classically hard
for multiplicative-error sampling
at interaction (maximum vertex) degree three,
whereas degree-two instances are exactly classically sampleable in
polynomial time for $p=O(\log n)$.
Thus the relevant degree
threshold depends essentially on whether the classical task is
evaluation of a local or global expectation value or sampling from
the complete QAOA output distribution.
Our reduction also concerns exact or exponentially precise algebraic evaluation; no claims are
made about computations outputting constant or polynomially small error.

Our main result is worst-case and does not imply hardness for all or for typical graphs (MaxCut problem instances).
An interesting future direction is a more precise characterization of which graph subfamilies our hardness obstruction applies to or not.
As mentioned, analytic or numerical solutions are possible in a number of cases. For example, a sequence of papers has obtained successively improved results concerning QAOA for MaxCut on high-girth graphs~\cite{Wang2018,hastings2019classical,wurtz2021maxcut,marwaha2021local,barak2021classical,basso2021quantum,sureshbabu2024parameter,farhi2025lower}, where our present reduction does not apply.
A particularly important open direction is to determine to what degree computational hardness of QAOA expectation values persists
at physically relevant levels of precision, in particular for constant or inverse-polynomial
additive error, as well as under physical noise~\cite{marshall2020characterizing,franca2021limitations,depalma2023limitations}.
Related numerical studies show that finite sampling and noise in objective function evaluations can place severe resource demands on variational optimization~\cite{scriva2024challenges,barligea2025scalability}. These results concern training and measurement costs rather than the classical complexity of expectation-value evaluation studied here. Resolving this question would help clarify the extent to which the
exact-complexity transition at $p=2$ constrains efficient quantum or classical evaluation
of QAOA expectation values, both theoretically and in practice. Similar considerations apply to related quantities such as their gradients and Hessians important for parameter optimization.

\paragraph{Implications for algorithms based on QAOA expectation values}
Recently a number of hybrid approaches have been proposed based on classical processing of quantum expectation values obtained using a quantum ansatz such as QAOA.
Examples include so-called Iterative Quantum Algorithms~\cite{bravyi2020obstacles,bravyi2020hybrid,dupont2023quantum,MIS,Finzgar2023,brady2025quantum,wybo2026scalable}, where correlators $\langle Z_iZ_j\rangle$ are iteratively estimated and processed to obtain a sequence of reduced problem instances, and related relaxation-based approaches~\cite{dupont2024extending,maciejewski2024multilevel}, where the
expectation values are classically processed to  produce candidate solutions directly.
Strengthening the implications of~\cite{WangEtAl2025}, Theorem~\ref{thm:main} shows that the expectation values appearing in these and other related approaches can be computationally hard to evaluate in the worst case, at least if exponential precision is required;
we point the reader to \cite[Sec. III.D]{wybo2026scalable} for additional relevant discussion. Our present results neither establish nor rule out potential quantum
advantage for these algorithms.

\paragraph{Connections to quantum supremacy experiments}
Our main result provides an expectation value analogue of existing sampling-hardness
results for low-depth QAOA  circuits~\cite{farhi2016quantum,harrow2017quantum,dalzell2020many}. Already at QAOA depth $p=2$, we show classical evaluation of a restricted QAOA for MaxCut expectation value suffices to recover
$\#\mathrm{P}$-hard counting information.
However, this does not by itself constitute an experimentally accessible
quantum-advantage result.  In particular, by Hoeffding's inequality, estimating
an edge correlator with measurement outcomes in $\{\pm1\}$
to additive error $\varepsilon$ using independent circuit
measurements can be achieved with
$O(\varepsilon^{-2}\log\frac{1}{\delta})$
independent samples to achieve failure probability at most $\delta$.
For the global MaxCut observable, a
computational basis measurement instead returns a cut value in
$[0,m]$, where $m=|E(G)|$.  Directly estimating the unnormalized
expectation $\langle C_G\rangle$ to additive error $\varepsilon$
can be achieved with $O(
    m^2\varepsilon^{-2}\log\frac{1}{\delta}
  )$
samples.
At the
accuracy $\varepsilon=2^{-\alpha N}$ used in our reduction,
these standard direct-sampling bounds are exponential; the additional
$m^2=\textrm{poly}(N)$ factor for the global objective doesn't change this
scaling.
Our result should consequently be interpreted as a structural
worst-case statement about QAOA
expectation value landscapes, rather than as one concerning
experimentally accessible precision.

\section*{Acknowledgment}
The author is grateful to Lucas Braydwood and Phillip Lotshaw for helpful discussions regarding the results of Ref.~\cite{WangEtAl2025}.
This material is based upon
work supported by the U.S. Department of Energy, Office of Science, Office of Advanced Scientific Computing Research under Award Number 89243024SSC000129, and under Prime Contract No. 80ARC020D0010 with the
NASA Ames Research Center.
The author acknowledges the use of ChatGPT-5.6-Sol AI tools as part of the drafting of Figure~\ref{fig:triangle-gadget} and for verification of the proofs of this manuscript. The author assumes complete responsibility for all reported findings.

\bibliographystyle{unsrt}
\nocite{lenstra1982factoring,basu2006algorithms,cohen2013course} \bibliography{bib}

\appendix

\section{Depth one QAOA expectation values are exactly tractable}\label{sec:p1}

For an edge $e=\{u,v\}$, let $d_u,d_v$ be the endpoint degrees and
$\lambda_{uv}=|N(u)\cap N(v)|$.
For each single edge term $C_{uv}=(I-Z_uZ_v)/2$
the standard depth-one QAOA  calculation~\cite{Wang2018,hadfield2018quantum} gives
\begin{align}
 \langle C_{uv}\rangle_{p=1}
 ={}&\frac12
 +\frac14\sin(4\beta)\sin\gamma
 \left(\cos^{d_u-1}\gamma+\cos^{d_v-1}\gamma\right)
 \notag\\
 &-\frac14\sin^2(2\beta)
 \cos^{d_u+d_v-2\lambda_{uv}-2}\gamma
 \left(1-\cos^{\lambda_{uv}}(2\gamma)\right),
 \label{eq:p1-formula}
\end{align}
with  $F_1(G;\boldsymbol\gamma,\boldsymbol\beta)
 = \sum_{\{u,v\}\in E} \langle C_{uv}\rangle_{p=1}$.
This formula follows by conjugating $Z_uZ_v$ by the mixer and phase operators, and retaining only Pauli
strings with nonzero expectation values for the initial state $|+\rangle^{\otimes n}$.
This result supplies the efficient depth one case of the main complexity dichotomy.
\begin{proposition}\label{prop:p1}
$\Eval^{uv}_1,\,\Eval^{\mathrm{MC}}_1\in\FP$.
\end{proposition}

\begin{proof}
Compute all degrees and common-neighbor counts, evaluate \eqref{eq:p1-formula} for every
edge, and sum, which requires a number of operations polynomially scaling in the problem size.  Under the exact input model specified in Appendix~\ref{app:single-query}, repeated squaring and exact field arithmetic evaluate the formula with polynomial bit complexity.
\end{proof}

Our results above show that, under standard complexity-theoretic assumptions, no analogous exact efficiently computable formulas are possible for arbitrary graphs with $p\geq 2$ in general.

\section{Single-query exact hardness with equal fixed mixers}
\label{app:single-query}

The proof of Theorem~\ref{thm:main} uses polynomially many values of the tied
phase angle to recover the Laurent coefficients.  For exact algebraic evaluation, the
same coefficients can instead be read from a single cyclotomic-field output.

For exact evaluation, we use the following explicit number field
representation.  An exact instance supplies a common
number field $K=\mathbb Q(\theta)$ by an irreducible defining polynomial and
isolating data identifying the chosen complex root $\theta$.  With $d=[K:\mathbb Q]$, each
value $a_0+a_1\theta+\cdots+a_{d-1}\theta^{d-1}$ is represented by its list of
rational basis coefficients $(a_0,\ldots,a_{d-1})$, encoded in binary.
The input supplies such lists for $\ii$, for each cost phase
$z_j=e^{\ii\gamma_j}$, and for each mixer pair
$(c_j,s_j)=(\cos\beta_j,\sin\beta_j)$.  The input is promised to satisfy
$|z_j|=1$, $c_j,s_j\in\mathbb R$, and $c_j^2+s_j^2=1$ under that embedding.
The exact oracle returns the coefficient list of the expectation value in the
same field and basis.  Input and output lengths are the respective total bit
lengths of these presentations and coefficient lists.  Other exact encodings
give the same conclusions whenever conversion to and from this representation
is polynomial time.  This qualification matters for the one-query statement below.
Write $\zeta_t=e^{2\pi\ii/t}$ for the primitive $t$-th root of unity.
The one-query construction uses the common field
$\mathbb Q(\zeta_8,\zeta_q)=\mathbb Q(\zeta_{8q})$, where $q$ is an odd
prime chosen in the proof.  This field has degree $4(q-1)$; its polynomial
bit-length bound is verified in the proof.

\par\Needspace{7\baselineskip}
\noindent With the representation above and equal fixed mixers, one exact oracle answer suffices for the counting reduction.\par\nobreak

\begin{theorem}[Single-query fixed-mixer hardness]
\label{thm:single-query-fixed-mixer}
Under the explicit number field representation above, for every fixed $p\geq2$,
exact evaluation of the MaxCut QAOA expectation is $\#\mathrm P$-hard under a
deterministic polynomial-time reduction making one oracle call.  This remains
true for the designated correlator
$\langle Z_rZ_s\rangle$ and for connected, unweighted, simple graphs of
diameter two.  At depth two the hard instances may be restricted to
\begin{equation}
 \beta_1=\beta_2=\frac{\pi}{4},
 \qquad
 \gamma_1=\gamma_2=\phi,
 \qquad
 e^{\ii\phi}=\zeta_q,
 \label{eq:single-query-angles}
\end{equation}
where $\zeta_q$ has polynomial order.  For $p>2$, all later angles may be set
to zero.
\end{theorem}

\begin{proof}
We use the explicit one-query hardness reduction for counting maximum cuts in
simple graphs given by Dell et al.~\cite[full version, Lemmas~C.2--C.3]{dell2014exponential}.  Given a source graph
$H$ with $n$ vertices, set
\begin{equation}
 \widehat H=
 \begin{cases}
 H,&n\text{ odd},\\
 H\sqcup K_1,&n\text{ even},
 \end{cases}
 \qquad
 \rho=|V(\widehat H)|-n\in\{0,1\}.
 \label{eq:single-query-padding}
\end{equation}
Then $\widehat n=|V(\widehat H)|$ is odd,
$\mc(\widehat H)=\mc(H)$, and
$|\M(\widehat H)|=2^\rho|\M(H)|$.  Apply the construction of
Section~\ref{sec:gadget} to $\widehat H$.  Because both $\widehat n$ and
$L$ are odd, the resulting graph $G=\mathcal G(\widehat H)$ satisfies
\begin{equation}
 N=2\widehat nL+2\equiv0\pmod4,
 \qquad
 |\M(G)|=2^{\rho+1}|\M(H)|.
 \label{eq:single-query-count}
\end{equation}
It retains all graph restrictions in the theorem.

Set $\beta_1=\beta_2=\pi/4$ and write the anchor correlator as
$H_{e_\star}(z)=\sum_{k=-2m}^{2m}a_kz^k$, where $m=|E(G)|$.
Theorem~\ref{thm:coefficient} and $N\equiv0\pmod4$ give
\begin{equation}
 a_{D_\star}=-|\M(G)|2^{2-2N}\ii^{N/2}\ne0.
 \label{eq:single-query-extreme}
\end{equation}
For the global expectation, Eq.~\eqref{eq:global-sum} gives the corresponding
coefficient $-a_{D_\star}/2$.  Thus either coefficient determines $|\M(G)|$.

All Laurent coefficients lie in $K_0=\mathbb Q(\zeta_8)$.  Bertrand's postulate
gives an odd prime $4m+1<q<2(4m+1)$, so $q=O(m)$.
Such a prime can be found in polynomial time by testing integers in this interval.
For the query field, take $K=\mathbb Q(\theta)$ with $\theta=\zeta_{8q}$
and cyclotomic defining polynomial $\Phi_{8q}$.
This polynomial has degree $4(q-1)$ and coefficients in $\{0,\pm1\}$,
and isolating data for the chosen root can be computed with polynomial bit complexity.
Encode the phases and mixers using $z=\theta^8$, $\ii=\theta^{2q}$,
and $c_j=s_j=(\theta^q+\theta^{7q})/2$ for $j=1,2$, reducing powers
modulo $\Phi_{8q}$.  These input lists have polynomial total bit length.
Query once at $z=\zeta_q$.  Since $\gcd(8,q)=1$,
\begin{equation}
 [K:K_0]=q-1>4m.
 \label{eq:single-query-degree}
\end{equation}
Consequently $1,\zeta_q,\ldots,\zeta_q^{4m}$ are linearly independent over
$K_0$, and
\begin{equation}
 \zeta_q^{2m}P(\zeta_q)
 =\sum_{j=0}^{4m}[z^{j-2m}]P(z)\,\zeta_q^j
 \label{eq:single-query-coordinates}
\end{equation}
for either $P=H_{e_\star}$ or the global expectation.
Multiply the returned value by $\zeta_q^{2m}$ and convert its rational coordinates
to the basis $\{\zeta_8^r\zeta_q^j:0\le r<4,\ 0\le j<q-1\}$ of $K$.
Grouping the four coordinates at each $j$ gives the $K_0$-coefficient
$[z^{j-2m}]P(z)$ for $0\le j\le4m$.  All remaining groups are zero.
This recovers every Laurent coefficient from the single answer.
The basis conversion uses polynomial-time exact linear algebra in degree
$4(q-1)=O(m)$.

For the bit-length bound, each path in Eq.~\eqref{eq:path-sum}
at these mixer angles contributes $2^{-2N}$ times one of $\{\pm1,\pm\ii\}$
to a Laurent coefficient.
There are at most $2^{3N}$ summands, so each Laurent coefficient of the
correlator has $O(N)$-bit rational coordinates in $\mathbb Q(\ii)$.
Summing over edges gives $O(N+\log m)$-bit coordinates for the global expectation.
Substitution of $\zeta_q$ and reduction in the cyclotomic power basis preserve
polynomial bit length.  Thus the exact answer and its decoded coefficients
have polynomial total length, and the decoding takes polynomial time.
The largest
exponent with a nonzero coefficient identifies $D_\star$.
Equations~\eqref{eq:single-query-extreme} and~\eqref{eq:single-query-count}
then recover $|\M(H)|$.  Zeroing later layers as before
proves the claim for $p>2$.
\end{proof}

\par\Needspace{7\baselineskip}
\noindent The final remark identifies the representation and precision limitations of the one-query proof.\par\nobreak

\begin{remark}
The one-query result uses the explicit number field representation above.
It also holds for exact encodings that allow polynomial-time conversion to and
from these basis coefficients, but is not claimed for arbitrary exact algebraic
encodings.  The proof uses exact outputs and gives no $O(N)$-bit replacement
for the Fourier argument in Section~\ref{sec:expoApprox}.  Whether
substantially higher precision permits algebraic reconstruction from a single
value is a separate question.  The decoding argument does not require padding.
With the main-text mixer choices $\tan\beta_1=1/2$ and $\beta_2=\pi/4$, the
coefficients lie in the fixed field
$\mathbb Q(\ii,\sqrt2,\sqrt5)\subseteq\mathbb Q(\zeta_{40})$.
Choosing $q$ coprime to $40$ gives the same one-query conclusion on the unpadded
graph $\mathcal G(H)$.  The padding above is used only to obtain equal mixer angles.
\end{remark}

\end{document}